\documentclass[11pt,aps,pra,onecolumn,superscriptaddress,floatfix,
nofootinbib,showpacs,longbibliography]{revtex4-2}
\usepackage[utf8]{inputenc}  
\usepackage[T1]{fontenc}     
\usepackage[british]{babel}  
\usepackage[sc,osf]{mathpazo}
\usepackage{palatino}  
\usepackage[colorlinks=true, citecolor=blue, urlcolor=blue]{hyperref}  
\usepackage{graphicx}
\usepackage{diagbox}
 \usepackage{enumitem}
 \usepackage[babel]{microtype}  
\usepackage{amsmath,amssymb,amsthm,bm,amsfonts,mathrsfs,bbm} 

\usepackage{xspace}  
\usepackage{pgf,tikz}
\usepackage{xcolor}
\usepackage{multirow}
\usepackage{array}
\usepackage{bigstrut}
\usepackage{braket}
\usepackage{color}
\usepackage{natbib}
\usepackage{multirow}
\usepackage{mathtools}
\usepackage{float}
\usepackage{xcolor,colortbl}
\usepackage{physics}
\usepackage{amsmath}
\usepackage{color}
\usepackage[justification=justified, format=plain]{subcaption}
\usepackage[justification=raggedright]{caption}

\newcommand{\be}{\begin{equation}}
\newcommand{\ee}{\end{equation}}
\newcommand{\ba}{\begin{eqnarray}}
\newcommand{\ea}{\end{eqnarray}}

\theoremstyle{plain}
\newtheorem{theorem}{Theorem}

\newtheorem{lemma}[theorem]{Lemma}

\newtheorem{proposition}[theorem]{Proposition}

\theoremstyle{definition}
\newtheorem{definition}[theorem]{Definition}

\theoremstyle{remark}
\newtheorem{remark}[theorem]{Remark}

\def\>{\rangle}
\def\<{\langle}

\begin{document}

\title{Quantum state exclusion with many copies}	

\author{Debanjan Roy}
\email{debanjanroy94@gmail.com, debanjan@jcbose.ac.in}
\affiliation{Department of Physical Sciences, Bose Institute, EN 80, Bidhannagar, Kolkata 700091, India}

\author{Tathagata Gupta\footnote{Present address: Department of Physics, Indian Institute of Technology Madras, Chennai 600036, India}}
\email{tathagatagupta@gmail.com}
\affiliation{Physics and Applied Mathematics Unit, Indian Statistical Institute, 203 B. T. Road, Kolkata 700108, India}

\author{Pratik Ghosal}
\email{ghoshal.pratik00@gmail.com}
\affiliation{S. N. Bose National Centre for Basic Sciences, Block JD, Sector III, Salt Lake, Kolkata 700 106, India}

\author{Samrat Sen}
\email{samrat9sen5@gmail.com}
\affiliation{Scuola Normale Superiore, Piazza dei Cavalieri 7, 56126 Pisa, Italy}

\author{Somshubhro Bandyopadhyay}
\email{som@jcbose.ac.in}
\affiliation{Department of Physical Sciences, Bose Institute, EN 80, Bidhannagar, Kolkata 700091, India}

\begin{abstract}
Quantum state exclusion is the task of identifying at least one state from a known set that was not used in the preparation of a quantum system. A set of quantum states is said to admit state exclusion if there exists a measurement whose outcomes can be put in one-to-one correspondence with the states in the set, such that each outcome rules out its corresponding state with certainty (while possibly also ruling out other states), and each outcome occurs with nonzero probability for at least one state in the set. State exclusion, however, is not always possible in the single-copy setting. In this paper, we investigate whether access to multiple identical copies of the system enables state exclusion. We prove that for any set of three or more pure states, state exclusion becomes possible with a finite number of copies. Moreover, we show that the number of copies required may be arbitrarily large: in particular, for every natural number $N$, we construct sets of states for which state exclusion remains impossible with $N$ or fewer copies.
\end{abstract}
\maketitle

\section{Introduction}
\label{sec:introduction}
Consider the following scenario. A referee prepares a quantum system in one of a finite number of states
\( \mathcal{S} = \{\rho_1, \rho_2, \ldots, \rho_k\} \), chosen according to some probability distribution, and hands the system to a player, Alice. Alice has complete knowledge of the set \( \mathcal{S} \), but she does not know which specific state has been prepared in any given run.

The most familiar task in this setting is \emph{quantum state discrimination} \cite{chefles2000quantum, bergou2010discrimination, barnett2009quantum, bae2015quantum}, where Alice attempts to identify the prepared state
by performing a suitable measurement on the system. In classical physics this task can always be carried
out perfectly, provided the possible states are distinct. Quantum mechanics, however, imposes a fundamental restriction: perfect discrimination is possible if and only if the states in \( \mathcal{S} \) are mutually orthogonal \cite{nielsen2010quantum}. As a result, whenever the states are nonorthogonal---as is typical in quantum theory---no measurement can identify the prepared state with certainty and zero error.

The impossibility of perfectly discriminating nonorthogonal quantum states motivates the study of alternative, less demanding information-processing tasks. The central idea behind these tasks is not to identify the prepared state exactly, but rather to extract some reliable information about it. One such task is \emph{quantum state exclusion} \cite{bandyopadhyay2014conclusive}. Here, instead of trying to determine the prepared state, the goal is more modest: identify at least one state in \( \mathcal{S} \) that the system was \emph{not} prepared in. Even this limited information can be useful, as excluding one or more possibilities reduces the set of candidate states consistent with the preparation. In this sense, state exclusion provides partial---but still operationally meaningful---information about the quantum state. As with state discrimination, prior knowledge of the set \( \mathcal{S} \) is essential to design measurements to optimally achieve the goal of exclusion. Although quantum state exclusion was introduced more recently than state discrimination, it has attracted significant attention due to its relevance in quantum foundations \cite{pusey2012reality}, quantum cryptography \cite{wallden2015quantum}, and operational approaches to quantum resources \cite{stratton2024operational,perry2015communication}.

While state exclusion is weaker than state discrimination, it must nevertheless be defined in a physically meaningful way. The basic idea is that a measurement used for state exclusion should always provide definite information about which states have not been prepared. We therefore say that a set of quantum states admits state exclusion if there exists a measurement whose outcomes can be put in one-to-one correspondence with the states in the set, such that observing a given outcome allows one to conclude with certainty that the corresponding state was not prepared (while the same outcome may also rule out other states). In addition, every outcome is required to be relevant, in the sense that it occurs with nonzero probability for at least one state in the set. Equivalently, every measurement outcome excludes at least one state with certainty, and every state in the set is excluded by at least one outcome. This notion may be regarded as perfect state exclusion.

To illustrate this idea, consider the following (unnormalized) set of three states in a three-dimensional Hilbert space:$\{\, |0\rangle + |1\rangle,\; |1\rangle + |2\rangle,\; |2\rangle + |0\rangle \,\}$. If Alice performs a projective measurement in the computational basis \(\{|0\rangle, |1\rangle, |2\rangle\}\), then each measurement outcome excludes exactly one of the states in the set. For example, obtaining the outcome ``\(0\)'' rules out the state \( |1\rangle + |2\rangle \), which has no overlap with \( |0\rangle \). In this way, the measurement achieves state exclusion. However, state exclusion is not always achievable. To see this, consider the following set of three qubit states:$\left\{
|0\rangle,\;
\frac{\sqrt{3}}{2}|0\rangle + \frac{1}{2}|1\rangle,\;
\frac{1}{\sqrt{2}}(|0\rangle + |1\rangle)
\right\}.$
In this case, no measurement can be constructed such that for every state one can associate an outcome that rules it out with certainty \cite{authorComment}.

The above examples show that state exclusion is possible for some sets of quantum states, while for others it is not. A set of quantum states is called \emph{antidistinguishable} if state exclusion, in the sense defined above, is possible for that set \cite{Heinosaari_2018}. Antidistinguishability therefore means that there exists a single measurement by which every state in the set can, in principle, be conclusively ruled out: for each state, there is a measurement outcome which tells us with certainty that this state was not prepared. When a set of states is not antidistinguishable, one may consider relaxed versions of the exclusion task. Some of the approaches include allowing a nonzero probability of error and optimising an appropriate figure of merit, or introducing an additional inconclusive measurement outcome that provides no exclusion but avoids erroneous conclusions. These generalised strategies allow us to quantify how well exclusion can be achieved when (perfect) state exclusion is impossible \cite{bandyopadhyay2014conclusive,mishra2024optimal}.

In this work, we ask whether access to multiple identical copies of a quantum system can enable state exclusion. We note that a related question has recently been considered in the context of a weaker formulation of the state exclusion problem, in which the requirement that every measurement outcome occurs with nonzero probability was not enforced \cite{ji2024barycentric}. In such a formulation, a measurement may be able to exclude only a subset of the states, and sets admitting this possibility are referred to as weakly antidistinguishable \cite{stratton2024operational}. In the present work, such sets are regarded as not antidistinguishable, as they violate the basic idea underlying antidistinguishability---namely, that every state in the set should, in principle, be excludable with certainty by an appropriate measurement outcome.

We now turn to the problem considered here. Suppose that a referee supplies Alice with several copies of a quantum system, all prepared in the same (unknown) state drawn from a known set $\mathcal{S}$ consisting of three or more states (for any set consisting of only two states, state discrimination and state exclusion are equivalent \cite{Heinosaari_2018}, and consequently providing any finite number of copies offers no advantage for state exclusion in this case). The question is whether states in $\mathcal{S}$ that are not antidistinguishable in the single-copy scenario can become antidistinguishable when Alice is given multiple copies of the system. This question is meaningful only when the number of copies is finite. With infinitely many copies, Alice could in principle reconstruct the prepared state via quantum state tomography \cite{d2003quantum}, rendering the task of state exclusion trivial.



The possibility that multiple copies might activate state exclusion is motivated by closely related phenomena in other quantum information–theoretic tasks. A notable example is unambiguous state discrimination, where inconclusive outcomes are permitted but errors are forbidden: it is known that any finite set of pure states becomes unambiguously distinguishable when a sufficiently large (yet finite) number of copies are available \cite{PhysRevA.64.062305}. Similar activation effects also arise in quantum channel discrimination, where any two unitary channels can be perfectly distinguished after finitely many uses \cite{acin2001statistical}.

However, such behaviour is not universal. In local entanglement transformation, for example, any pair of incomparable pure states in $3 \times 3$ remains incomparable for any finite number of copies \cite{bandyopadhyay2002classification}. These contrasting examples show that access to many copies may or may not help in an information processing task. The goal of this work is to determine which of these two behaviours applies to quantum state exclusion.

In this paper we prove the following results:
\begin{itemize}
\item Every set of $k\geq3$ pure states that is not antidistinguishable becomes antidistinguishable with a finite number of copies.

In other words, antidistinguishability is universally activatable in the many-copy regime for pure-state ensembles of size at least three. The number of copies required depends on the set under consideration. For specific classes of pure state sets---namely, sets of three or more states with equal and real pairwise inner products, and sets of three states with equal pairwise overlaps---we determine this number exactly.
\end{itemize}
Our next result demonstrates that the number of copies required may be arbitrarily large. \begin{itemize}
\item For any fixed positive integer $N$, there exist sets of pure states that are not antidistinguishable with $N$ or fewer copies. In particular, for each $N$, we explicitly construct sets of states that fail to be antidistinguishable with $N$ or fewer copies but become antidistinguishable with $N+1$ copies.
\end{itemize} 

One may observe that our result is very similar to a fundamental result in entanglement distillation: for every positive integer $N$, there exist states that are distillable but not $N$-distillable \cite{watrous2004many}. 

\section{Exclusion of quantum states}
\label{prelim}
In a state exclusion problem, we assume a referee prepares a quantum system in a state $\rho_i$ chosen from the set $\mathcal{S} = \{\rho_1, \rho_2, \ldots, \rho_k\}$ with probability $p_i$ and hands the system to Alice. Alice has complete knowledge of $\mathcal{S}$ but not which state was prepared—her goal is to rule out at least one state that the system was not prepared in. Concretely, she must output a label $j \in \{1,2,\ldots,k\}$ such that $\rho_j \neq \rho_i$. If she can always do this with certainty, we say that the set $\mathcal{S}$ admits state exclusion.

\subsection{Antidistinguishable Sets}
\begin{definition}\label{def:antidistinguishable}
A set $\mathcal{S} = \{\rho_1, \rho_2, \ldots, \rho_k\}$ is \emph{antidistinguishable} if and only if there exists a measurement (POVM) $\mathcal{M} = \{\Pi_j\}_{j=1}^k$ such that
\begin{equation}
\operatorname{Tr}(\rho_j \Pi_j) = 0, \qquad \forall j \in \{1,\ldots,k\},
\label{eq:antidist1}
\end{equation}
and
\begin{equation}
\sum_{i=1}^k \operatorname{Tr}(\rho_i \Pi_j) > 0, \qquad \forall j \in \{1,\ldots,k\}.
\label{eq:antidist2}
\end{equation}
\end{definition}

The meaning of these conditions is straightforward:
\begin{itemize}
    \item Equation~\eqref{eq:antidist1} ensures exclusion: if outcome $j$ occurs, Alice knows
    with certainty that the system was not prepared in state $\rho_j$.
    \item Equation~\eqref{eq:antidist2} implies that all measurement outcomes are relevant—each
    outcome occurs with nonzero probability for at least one state in $\mathcal{S}$. Therefore, all $k$
    outcomes are actually used across different runs of the experiment.
\end{itemize}

Note that if state exclusion is possible, then for every state in the set, there exists at least one measurement outcome that excludes it. This, however, does not imply that each outcome excludes \emph{exactly} one state. An outcome may exclude more than one state; nevertheless, there must exist an assignment of outcomes to states such that each state is excluded by a distinct outcome, satisfying Eq.~\eqref{eq:antidist1}.

The notion of antidistinguishability used in Definition \ref{def:antidistinguishable} provides a natural operational formulation of state exclusion and has been referred to in the literature as \emph{strong antidistinguishability} \cite{stratton2024operational}. While condition \eqref{eq:antidist1} is clearly essential for exclusion, condition \eqref{eq:antidist2} is equally crucial: it ensures that all measurement outcomes are relevant, in the sense that none is redundant or never observed (occurs with zero probability for every state in the set).

A weaker notion, briefly mentioned earlier in the introduction, is referred to as \emph{weak antidistinguishability} \cite{stratton2024operational}, where condition \eqref{eq:antidist1} is required but not \eqref{eq:antidist2}. That is, no requirement is imposed that all outcomes occur with nonzero probability for at least one state. Consequently, some outcomes may be irrelevant or never observed, and the corresponding measurement may be able to exclude only a proper subset of the states in the set. Operationally, this means that while certain states can be conclusively ruled out, others can never be excluded in any experimental run. For this reason, sets that are only weakly antidistinguishable will be regarded as not antidistinguishable.

One may consider variations of the exclusion problem, such as $m$-state exclusion, where each outcome must rule out $m$ candidate states \cite{stratton2024operational}.

\subsection{Known Characterizations}
Unlike state discrimination, a complete characterization of antidistinguishable sets is
not known. Only partial results are available.

\subsubsection*{Necessary condition}
\begin{proposition}[from \cite{bandyopadhyay2014conclusive}]\label{necessary}
A set $\mathcal{S} = \{\rho_1, \rho_2, \ldots, \rho_k\}$ is antidistinguishable \emph{only if}
\begin{equation}
\sum_{i<j=1}^k F(\rho_i,\rho_j) \le \frac{k(k-2)}{2},
\end{equation}
where
\begin{equation}
F(\rho_i,\rho_j) = \operatorname{Tr}\!\left(\sqrt{\sqrt{\rho_i}\,\rho_j\,\sqrt{\rho_i}}\right)
\end{equation}
denotes the fidelity between $\rho_i$ and $\rho_j$.
\end{proposition}

If this inequality is violated, the set cannot be antidistinguishable. Note that a set satisfying
the above inequality may still not be antidistinguishable as the condition is only necessary but not
sufficient. For such an example, consider the following set of three qubit states: $\mathcal{S}=\{\ket{0},\ket{1},\frac{1}{2}\ket{0}+\frac{\sqrt{3}}{2}\ket{1}\}$. Although this set satisfies the necessary condition, it is not antidistinguishable \cite{authorComment}. A sufficient condition, however, is known for pure states.

\subsubsection*{Sufficient condition for pure states}

\begin{proposition}[from \cite{Heinosaari_2018}]\label{sufficient}
Let $\mathcal{S} = \{|\psi_1\rangle, |\psi_2\rangle, \ldots, |\psi_k\rangle\}$ be a set of pure states. The set
is antidistinguishable \emph{if} there exist positive coefficients $t_1,\ldots,t_k$ such that
\begin{equation}
\sum_{i=1}^k t_i\,|\psi_i\rangle\!\langle\psi_i| = \mathbb{I}.
\end{equation}
\end{proposition}

For qubits, this condition is both necessary and sufficient. In higher dimensions, it remains only a sufficient condition.

\subsubsection*{Necessary and sufficient condition for pure states in higher dimension}

For general Hilbert-space dimension $d$, a necessary and sufficient condition for the antidistinguishability of pure states is known and can be expressed in terms of the Gram matrix of the set.

\begin{proposition}[from \cite{Johnston2025tightbounds}]\label{prop:incoherence_cond}
Let $\mathcal{S} = \{|\psi_1\rangle, |\psi_2\rangle, \ldots, |\psi_k\rangle\} \subset \mathbb{C}^d$ be a set of
$k$ pure states, and let $G \in \mathrm{Pos}(\mathbb{C}^k)$ denote the Gram matrix of $\mathcal{S}$ with entries
$G_{ij} = \langle\psi_i|\psi_j\rangle$, where $\mathrm{Pos}(\mathbb{C}^k)$ denotes the set of positive
semidefinite operators on $\mathbb{C}^k$. The set $\mathcal{S}$ is antidistinguishable \emph{if and only if} $G$ is
$(k-1)$-incoherent, i.e., $G$ admits a decomposition
\begin{equation}
G = \sum_{i=1}^k F_i, \qquad F_i \in \mathrm{Pos}(\mathbb{C}^k),
\end{equation}
such that
\begin{equation}
\langle i|F_i|i\rangle = 0, \qquad \forall i \in \{1,\ldots,k\}.
\end{equation}
\end{proposition}

For a given set of pure states, the condition in Proposition \ref{prop:incoherence_cond} can be efficiently checked using semidefinite programming.

\subsubsection*{Three pure states}

In the special case of three pure states, a simpler and more explicit necessary and sufficient condition is known.

\begin{proposition}[from \cite{caves2002conditions}]\label{prop:three-state-iff}
Let $\mathcal{S} = \{|\psi_1\rangle, |\psi_2\rangle, |\psi_3\rangle\}$. The set is antidistinguishable \emph{if and only if}
\begin{equation}
\sum_{i<j=1}^3 x_{ij} < 1,
\end{equation}
and
\begin{equation}
\left(\sum_{i<j=1}^3 x_{ij} - 1\right)^2 \ge 4\prod_{i<j=1}^3 x_{ij},
\end{equation}
where $x_{ij} \coloneqq |\langle\psi_i|\psi_j\rangle|^2$.
\end{proposition}
    
\section{Many-copy Antidistinguishability }\label{results}
\subsection{Motivation}

We now turn to sets of states that are not antidistinguishable in the single-copy scenario. For such sets, it is natural to ask whether access to multiple identical copies of the prepared state can enable exclusion. Consider the set of states
\begin{equation}\label{weak}
\mathcal{S} = \{\,|0\rangle, |+\rangle, |1\rangle\,\},
\end{equation}
where $|+\rangle = \tfrac{1}{\sqrt{2}}(|0\rangle + |1\rangle)$. From Proposition \ref{sufficient}, as well as Proposition \ref{prop:three-state-iff}, it follows that this set is not antidistinguishable.. However, when two copies are provided, the corresponding set
\begin{equation}
\mathcal{S}_2 =
\left\{
|0\rangle^{\otimes 2},\;
|+\rangle^{\otimes 2},\;
|1\rangle^{\otimes 2}
\right\}
\end{equation}
becomes antidistinguishable. Indeed, one can verify that $\mathcal{S}_2$ satisfies both conditions of
Proposition \ref{prop:three-state-iff}. An explicit joint measurement, a $3$-outcome POVM $\mathcal{M} = \{\Pi_i\}_{i=1}^3$,
acting on the two-copy Hilbert space $\mathbb{C}^2 \otimes \mathbb{C}^2$ and satisfying
Eqs.~\eqref{eq:antidist1} and~\eqref{eq:antidist2} is given below: 
\begin{equation} 
\begin{aligned}
    \Pi_1 &=
\begin{pmatrix}
0      & 0      & 0      & 0      \\
0      & 0.3671 & 0.0338 & 0.0991 \\
0      & 0.0338 & 0.3671 & 0.0991 \\
0      & 0.0991 & 0.0991 & 0.8017
\end{pmatrix}, \\
\Pi_2 &= \begin{pmatrix}
0.1983 & -0.0991 & -0.0991 & 0      \\
-0.0991 & 0.2658 & -0.0675 & -0.0991 \\
-0.0991 & -0.0675 & 0.2658 & -0.0991 \\
0      & -0.0991 & -0.0991 & 0.1983
\end{pmatrix}, \\
\Pi_3 &= \begin{pmatrix}
0.8017 & 0.0991 & 0.0991 & 0      \\
0.0991 & 0.3671 & 0.0338 & 0      \\
0.0991 & 0.0338 & 0.3671 & 0      \\
0      & 0      & 0      & 0
\end{pmatrix}.
\end{aligned} 
\end{equation} 
\subsection{Results}
In the above example, a joint measurement on two copies serves the purpose of state exclusion. This suggests that providing more copies should only make exclusion “easier”, which is also intuitively satisfying. The question, however, is whether this intuition always holds in a strong sense: \emph{Does every non-antidistinguishable set become antidistinguishable using only finitely many copies, or do there exist sets for which exclusion remains impossible regardless of how many copies are supplied?}

Our first result shows that any set of three or more pure states is antidistinguishable with finitely many copies.  

\subsubsection{All sets of pure states are many-copy antidistinguishable}

Before stating our first main theorem, we recall a useful sufficient condition based on pairwise
overlaps.

\begin{lemma}[from \cite{Johnston2025tightbounds}]\label{lemma:sufficient}
Let $\mathcal{S} = \{|\psi_1\rangle, |\psi_2\rangle, \ldots, |\psi_k\rangle\}$ be a set of pure states. The set
is antidistinguishable \emph{if}
\begin{equation*}
\bigl|\langle\psi_i|\psi_j\rangle\bigr|
\le \frac{1}{\sqrt{2}}
\sqrt{\frac{k-2}{k-1}},
\qquad \forall\,i \neq j.
\end{equation*}
\end{lemma}
Lemma \ref{lemma:sufficient} tells us that if all pairwise overlaps are below a threshold that depends only on $k$, then the set is guaranteed to be antidistinguishable. Although a given set may fail to satisfy this condition in the single-copy setting, in the many-copy setting, it may not because all overlaps decrease simultaneously with the number of copies. Consequently, one expects that for sufficiently many copies the overlaps will eventually fall below the threshold in Lemma \ref{lemma:sufficient}, thereby ensuring antidistinguishability. This leads to first main result---universal ``activation'' of antidistinguishability.

\begin{theorem}\label{every_set}
Every set of $k\geq 3$ pure states that is not antidistinguishable becomes antidistinguishable with a finite number of copies.

\end{theorem}

\begin{proof}
    Let $\mathcal{S}=\{\ket{\psi_1},\ket{\psi_2},\ldots,\ket{\psi_k}\}$ be a set of $k\geq 3$ pure states, and define
    \begin{align}
        c \coloneq \max_{i\neq j} \abs{\bra{\psi_i}\ket{\psi_j}} \in [0,1).
    \end{align} If $c \leq \frac{1}{\sqrt{2}}\sqrt{\frac{k-2}{k-1}}$, then $\mathcal{S}$ is already antidistinguishable (single copy), by Lemma~\ref{lemma:sufficient}.

    Suppose instead that
    \begin{align}
        c > \frac{1}{\sqrt{2}}\sqrt{\frac{k-2}{k-1}}
    \end{align} and consider the $N$-copy exclusion problem corresponding to the set $\mathcal{S}_N = \left\{\ket{\psi_1}^{\otimes N},\ldots,\ket{\psi_k}^{\otimes N}\right\}$. The pairwise overlaps satisfy
    \begin{align}
            \abs{\braket{\psi_i^{\otimes N}}{\psi_j^{\otimes N}}}=\abs{\braket{\psi_i}{\psi_j}}^N \leq c^N.
    \end{align}
    Since $0\le c<1$, the sequence $c^N$ decreases monotonically with $N$. In particular, if
    \begin{align}
        N \ge \left\lceil \frac{\ln t_k}{\ln c} \right\rceil,
        \quad \text{with }
        t_k = \frac{1}{\sqrt{2}}\sqrt{\frac{k-2}{k-1}}, \label{eq:bound}
    \end{align}
    then $c^N \le t_k$, and hence
    \begin{align}
        \abs{\bra{\psi_i^{\otimes N}}\ket{\psi_j^{\otimes N}}} \leq t_k, \quad \forall\, i\neq j.
    \end{align}
    By Lemma~\ref{lemma:sufficient}, the set $\mathcal{S}_N$ is antidistinguishable. This completes the proof.
\end{proof}

This result holds for any set of cardinality $k\geq 3$ and for any Hilbert space dimension $d$.

Notably, Eq.~\eqref{eq:bound} provides an upper bound on the number of copies required for a given set to be many-copy antidistinguishable. However, since this bound follows from a sufficient condition, Theorem~\ref{every_set} does not specify the minimum number of copies for which a given set $\mathcal{S}$ becomes antidistinguishable. This motivates the following definition:
\begin{definition}[$N$-copy antidistinguishability]\label{N-copy-AD}
    Fix a natural number $N$. A set of quantum states $\mathcal{S} = \{\rho_1, \rho_2, \ldots, \rho_k\}$ is said to be \emph{$N$-copy antidistinguishable} if it is antidistinguishable with $N$ copies but not fewer.
\end{definition}
The above definition implies that if $\mathcal{S}$ is $N$-copy antidistinguishable, then the set $\mathcal{S}_{N'}=\{\rho_1^{\otimes N'}, \rho_2^{\otimes N'}, \ldots, \rho_k^{\otimes N'}\}$ is not antidistinguishable for any $N' < N$. It is also easy to see that an $N$-copy antidistinguishable set is also antidistinguishable with $N'$ copies, for all $N' \ge N$. In the following, we show that for certain classes of state sets---such as those with equal and real pairwise inner products--- the required number of copies $N$ can be determined exactly. As expected, $N$ depends both on the inner product and on the number of states in the set.

\subsubsection{$N$-copy antidistinguishability of states with equal and real pairwise inner products}

We first recall the following result.

\begin{lemma}[from \cite{Johnston2025tightbounds}]\label{neccesary_and_sufficient}
    Let $\mathcal{S}=\{\ket{\psi_1},\ket{\psi_2},\cdots,\ket{\psi_k}\}$ be a set of pure states satisfying
    \begin{align}
        \braket{\psi_i}{\psi_j}=\gamma,\quad \forall~i\neq j,
    \end{align}
    where $0\leq\gamma<1$. The set $\mathcal{S}$ is antidistinguishable \emph{if and only if}
    \begin{align}
        \gamma \leq \frac{k-2}{k-1}.
    \end{align}
\end{lemma}

Let us now suppose that the given set $\mathcal{S}$ is not antidistinguishable, which implies that for such a set $\gamma > \frac{k-2}{k-1}$. We now derive the exact minimum number of copies required for $\mathcal{S}$ to become many-copy antidistinguishable.

\begin{theorem}\label{tight}
Let $\mathcal{S} = \{ \ket{\psi_1}, \ket{\psi_2}, \ldots, \ket{\psi_k} \}$ be a set of $k \ge 3$ pure states
with equal and real inner product $\gamma\in \left(\frac{k-2}{k-1},1\right)$.
The set $\mathcal{S}$ is $N$-copy antidistinguishable for
\begin{equation}
N = \left\lceil \frac{\ln s_k}{\ln \gamma} \right\rceil ,
\quad \text{where } s_k := \frac{k-2}{k-1}.
\end{equation}
\end{theorem}

\begin{proof}
Consider the $N'$-copy set
\begin{equation}
\mathcal{S}_{N'} =
\left\{
\ket{\psi_1}^{\otimes N'},
\ket{\psi_2}^{\otimes N'},
\cdots,
\ket{\psi_k}^{\otimes N'}
\right\},
\quad \text{where } N' \ge 2.
\end{equation}
The pairwise inner products of the $N'$-copy states satisfy
\begin{equation}
\braket{\psi_i^{\otimes N'}}{\psi_j^{\otimes N'}}
= \braket{\psi_i}{\psi_j}^{N'}
= \gamma^{N'}
\quad \forall\, i \ne j.
\end{equation}
Since $\gamma \in \left(\frac{k-2}{k-1},\,1\right)$, it follows that
\begin{equation}
\gamma^{N'} \le \frac{k-2}{k-1}
\quad \text{whenever} \quad
N' \ge N = \left\lceil \frac{\ln s_k}{\ln \gamma} \right\rceil.
\end{equation}
Therefore, by Lemma~\ref{neccesary_and_sufficient}, the set $\mathcal{S}_{N'}$ is antidistinguishable for all
$N' \ge N$.

On the other hand, for any $N' < N$, we have $\gamma^{N'} > \frac{k-2}{k-1}$, and hence
Lemma~\ref{neccesary_and_sufficient} implies that $\mathcal{S}_{N'}$ is not antidistinguishable. This completes the proof.
\end{proof}

\begin{remark}
Since $t_k := \sqrt{\frac{k-2}{2(k-1)}} < s_k$ and $\ln \gamma < 0$, it follows that
$\left\lceil \frac{\ln s_k}{\ln \gamma} \right\rceil
\le
\left\lceil \frac{\ln t_k}{\ln \gamma} \right\rceil$.
Hence, for sets with equal and real inner products, antidistinguishability
can be activated using strictly fewer copies than those predicted by the general
bound in Eq.~\eqref{eq:bound}, which is also necessary.
\end{remark}

\subsubsection{Arbitrarily many copies may be required for activating antidistinguishability}

Theorem~\ref{every_set} guarantees that antidistinguishability can always be activated
given a sufficient number of copies. This naturally raises the question of
whether there exists a finite number $N$ such that any set of three or more
pure states becomes antidistinguishable with at most $N$ copies. The
following theorem answers this question in the negative.

We begin with a lemma from Ref.~\cite{Johnston2025tightbounds}, followed by our theorem.

\begin{lemma}[from \cite{Johnston2025tightbounds}]\label{necessary_lemma}
    A set $\mathcal{S}=\{\ket{\psi_1},\cdots,\ket{\psi_k}\}$ of $k\geq 3$ pure states is \emph{not antidistinguishable} if 
    \begin{align}
       \abs{\bra{\psi_i}\ket{\psi_j}}  > \frac{k-2}{k-1}, \quad \forall~i\neq j.
    \end{align}
\end{lemma}

\begin{theorem}\label{every_N}
    For every natural number $N$, there exist sets of pure states that are not $N$-copy antidistinguishable.
\end{theorem}

\begin{proof}
   
  Fix $N\in \mathbb{N}$  and consider a set $\mathcal{S}=\{\ket{\psi_1},\cdots,\ket{\psi_k}\} \subset \mathbb{C}^d$ of $k\geq3$ states with equal pairwise overlaps
    \begin{align} \label{eq:equal-overlap}
\abs{\bra{\psi_i}\ket{\psi_j}}&=\gamma:=\frac{1+a}{2}, \quad i\neq j\\
        \text{where}\quad a&:=\left(\frac{k-2}{k-1}\right)^{1/N}.
    \end{align}
    Such sets exist whenever $d\geq k$, because for these dimensions there exist sets of $k$ pure states with equal pairwise overlaps $\gamma$ for any $\gamma\in(0,1)$ \cite{sustik2007existence, waldron2018introduction} (See also Remark \ref{remark}).
    
Since $0<a<1$, and  $a<\gamma<1$, we have
    \begin{align}
        \gamma^N>a^N = \frac{k-2}{k-1}.
    \end{align}

  For the $N'$-copy set $\mathcal{S}_{N'}=\left\{\ket{\psi_1}^{\otimes N'},\cdots,\ket{\psi_k}^{\otimes N'}\right\}$, the pairwise overlaps satisfy
    \begin{align}
        \abs{\bra{\psi_i^{\otimes N'}}\ket{\psi_j^{\otimes N'}}}
        =\abs{\bra{\psi_i}\ket{\psi_j}}^{N'}
        =\gamma^{N'}
    \end{align}
    and therefore for any $N'\leq N$,
    \begin{align}
        \gamma^{N'}\geq \gamma^{N}> \frac{k-2}{k-1}.
    \end{align}
    
  By Lemma~\ref{necessary_lemma}, $\mathcal{S}_{N'}$ is therefore not antidistinguishable for any $N' \leq N$.
\end{proof}

\begin{remark}\label{remark}
     For $k > d$, the Welch bound \cite{welch2003lower} imposes a lower bound on $\gamma$ through the relation \begin{equation*}
        \gamma^2 \geq \frac{k-d}{d(k-1)}
    \end{equation*} although the question of existence of such sets of equiangular lines for all $d$ is not yet resolved. The Gerzon's bound imposes an upper bound on $k$ by stating that $k$ equiangular lines \cite{renes2004symmetric, scott2010symmetric, appleby2007symmetric} in $\mathbb{C}^d$ must satisfy $k\leq d^2$ \cite{lemmens1973equiangular}. Recently a constructive proof for existence of such sets in the $k=d^2$ case (known as the Zauner's conjecture) has been proposed \cite{appleby2025constructive}. 
\end{remark} 

So far we have seen that a set of pure states, assuming such a set is not single-copy antidistinguishable, nevertheless becomes antindistinguishable with sufficiently many copies---and that the number of copies, though finite, could be arbitrarily large. Our next result in some sense brings these two results together. We show that for every natural number $N$ there exist sets of states for all values of cardinality $\geq3$ that are $(N+1)$-copy antidistinguishable. 

\begin{theorem}\label{N_N+1}
    For every positive integer $N$, there exist sets of pure states that are $(N + 1)$-copy antidistinguishable.
\end{theorem}

\begin{proof}
   The proof is constructive. Fix a positive integer $N$, and consider a set
$\mathcal{S}=\{\ket{\psi_1},\allowbreak\ket{\psi_2},\allowbreak\cdots,\allowbreak\ket{\psi_k}\}\subset \mathbb{C}^d$
of $k\geq 3$ pure states with equal and real pairwise inner products

\begin{align}
        \braket{\psi_i}{\psi_j}
        = \gamma
        :=\left(\frac{k-2}{k-1}\right)^{\frac{2N+1}{2N(N+1)}},
        \quad \forall~i\neq j.
    \end{align}
    The existence of such a set in dimension $d\geq k$ follows from the fact that, for any $k\geq 2$, any $d\geq k$, and any $\gamma\in[0,1)$, there exist $k$ pure states in $\mathbb{C}^d$ with equal and real pairwise inner product $\gamma$ \cite{roa2011conclusive}.

    Next, observe that
    \begin{align}
        \frac{1}{N+1} < \frac{2N+1}{2N(N+1)} < \frac{1}{N}.
    \end{align}
    Since $0<\frac{k-2}{k-1}<1$, it follows that
    \begin{align}
        \left(\frac{k-2}{k-1}\right)^{\frac{1}{N}} < \gamma < \left(\frac{k-2}{k-1}\right)^{\frac{1}{N+1}},
    \end{align}
    and hence
    \begin{align}
        \gamma^{N} > \frac{k-2}{k-1}\\\gamma^{N+1} < \frac{k-2}{k-1}.
    \end{align}
    
    Therefore, by Lemma \ref{neccesary_and_sufficient}, the $N$-copy set $\mathcal{S}_N$ is not antidistinguishable, whereas the $(N+1)$-copy set $\mathcal{S}_{(N+1)}$ is antidistinguishable. This completes the proof.  
    \end{proof}

\subsubsection{Complete characterisation of $N$-copy antidistinguishability of three pure states with equal overlaps}

The above construction (and also Theorem \ref{tight}) provides a complete and explicit characterisation of the number of copies necessary and sufficient to activate antidistinguishability for sets of three or more pure states with equal and real pairwise inner products, in terms of the inner product. However, in the special case of three pure states, one can obtain a similar characterization even when the  inner products are allowed to take complex values.

\begin{proposition}\label{prop:three}
    Let $\mathcal{S}=\{\ket{\psi_1},\ket{\psi_2},\ket{\psi_3}\}$ be a set of three pure states with equal pairwise overlaps and let $x=\abs{\braket{\psi_i}{\psi_j}}^2,~~\forall~ i\neq j,$ be the squared pairwise overlap. The set $\mathcal{S}$ is $N$-copy antidistinguishable \emph{if and only if}
    \begin{align}
        \left(\frac{1}{4}\right)^{\frac{1}{N-1}} <x \le \left(\frac{1}{4}\right)^{\frac{1}{N}}.
    \end{align}
\end{proposition}

\begin{proof}
    We begin by first noting that from Proposition~\ref{prop:three-state-iff}, it follows that $\mathcal{S}$ is antidistinguishable if and only if
    \begin{align}
        &x < \frac{1}{3},\label{c1} \\
        \text{and}~~&(3x-1)^2\geq 4x^3.\label{c2}
    \end{align}
    However, these two conditions are not independent; satisfaction of Eq.~\eqref{c2} imply the satisfaction of Eq.~\eqref{c1}. Note that it may happen that for some values of $x$, Eq.~\eqref{c1} is satisfied but Eq.~\eqref{c2} is not. Therefore, we can say that $\mathcal{S}$ is antidistinguishable \emph{if and only if} Eq.~\eqref{c2} is satisfied, i.e., $0\leq x\leq\frac{1}{4}$.  

    For $x>\frac{1}{4}$, the set $\mathcal{S}$ is therefore not antidistinguishable. We then consider the $N$-copy set $\mathcal{S}_N=\{\ket{\psi_1}^{\otimes N},\ket{\psi_2}^{\otimes N},\ket{\psi_3}^{\otimes N}\}$. The pairwise squared overlaps of the $N$-copy states are given by
    \begin{align}
        y:=\abs{\braket{\psi_i^{\otimes N}}{\psi_j^{\otimes N}}}^2=\abs{\braket{\psi_i}{\psi_j}}^{2N}=x^N,\quad\forall i\neq j\in\{1,2,3\}.
    \end{align}
    Therefore, $\mathcal{S}_N$ is antidistinguishable if and only if 
    \begin{align}
        &(3y-1)^2-4y^3\geq 0,\nonumber\\
        \text{or,}\quad&y\leq \frac{1}{4} \iff x\leq \left(\frac{1}{4}\right)^{\frac{1}{N}}.
    \end{align}

    On the other hand, the $(N-1)$-copy set $\mathcal{S}_{(N-1)}=\{\ket{\psi_1}^{\otimes (N-1)},\ket{\psi_2}^{\otimes (N-1)},\ket{\psi_3}^{\otimes (N-1)}\}$ is not antidistinguishable if and only if
    \begin{align}
        y'=x^{N-1} > \frac{1}{4} \iff x > \left( \frac{1}{4}\right)^{\frac{1}{N-1}},
    \end{align}
    where $y':=\abs{\braket{\psi_i^{\otimes (N-1)}}{\psi_j^{\otimes (N-1)}}}^2$ denotes the pairwise squared overlaps of the $(N-1)$-copy states.

    Therefore the set $\mathcal{S}$ is $N$-copy antidistinguishable if and only if $\left( \frac{1}{4}\right)^{\frac{1}{N-1}}<x\leq\left(\frac{1}{4}\right)^{\frac{1}{N}}$. This completes the proof.
\end{proof}




\begin{figure}[t!]
\centering
\includegraphics[scale=0.65]{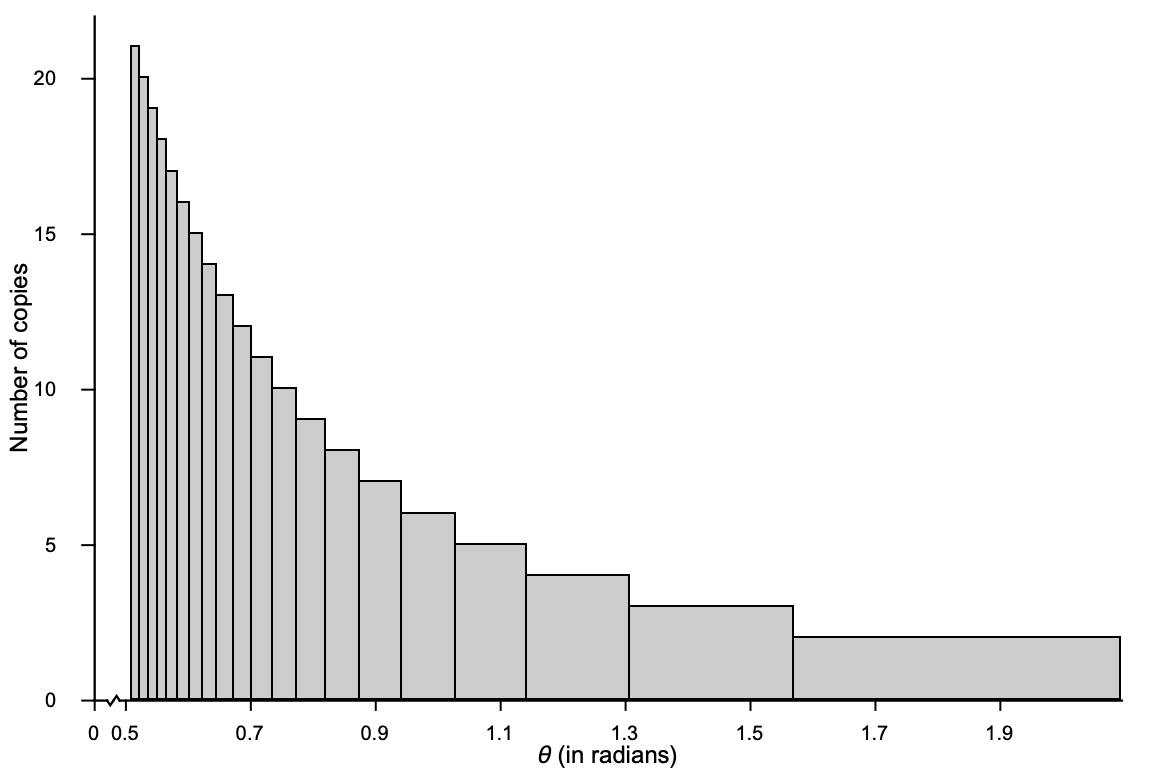}
\caption{
Number of copies necessary and sufficient to activate antidistinguishability for the set of three qubit states defined in Eq.~\eqref{ex}, as a function of the parameter $\theta$. A line break has been used as the required number of copies becomes arbitrarily large as $\theta$ approaches zero.}
\label{fig1}
\vspace{-0.5cm}
\end{figure}

\noindent\textbf{Example.}
To illustrate the above result, we consider the following family of three qubit states:
\begin{align}
    \mathcal{S}
    =
    \Bigl\{
        \ket{0},\,
        \cos\frac{\theta}{2}\ket{0}+\sin\frac{\theta}{2}\ket{1},\,
        \cos\frac{\theta}{2}\ket{0}+e^{i\phi}\sin\frac{\theta}{2}\ket{1}
    \Bigr\},\nonumber\\
    \text{where,}\qquad \phi=\arccos\!\left(\frac{\cos\theta}{1+\cos\theta}\right),
\quad\text{and}\quad
\theta\in\left[0,\frac{2\pi}{3}\right].\label{ex}
\end{align}
For this choice of parameters, the three states have equal pairwise overlaps, with
\(
\abs{\braket{\psi_i}{\psi_j}}^2=\cos^2\left({\frac{\theta}{2}}\right)
\)
for all $i\neq j$.

Using the criterion derived above, one can determine the necessary and sufficient number of copies required to activate antidistinguishability as a function of the parameter $\theta$. This dependence is shown in Fig.~\ref{fig1}, which shows how the required copy number increases exponentially as the states become less distinguishable (i.e., as $\theta$ decreases, or $x=\cos^2\left(\frac{\theta}{2}\right)$ increases).

\section{Conclusions}
\label{conclusion}

In conclusion, we investigated the role of multiple identical copies in enabling quantum state exclusion for sets of states that are not antidistinguishable in the single-copy setting.  We proved that any set of three or more pure states becomes antidistinguishable with a finite number of copies [Theorem \ref{every_set}]. This result is similar to the well-known result in unambiguous discrimination of pure states, where sufficient copies guarantee unambiguous distinguishability.

The number of copies required to activate state exclusion, however, depends on the state set. We derived a general upper bound in terms of the set cardinality and the maximum pairwise overlap of the states [Eq.~\eqref{eq:bound}  in the proof of Theorem \ref{every_set}], although this bound is not tight in general. For the special class of pure-state sets with equal and real inner products, we determined the exact minimum number of copies and showed that it is strictly smaller than the general bound [Theorem~\ref{tight}].

While the required number of copies is always finite, it can be arbitrarily large. In particular, for every positive integer $N$, we constructed sets of pure states that are not $N'$-copy antidistinguishable for any $N' \leq N$ [Theorem~\ref{every_N}]. We also presented a class of state sets which remains non-antidistinguishable with $N$ or fewer copies, but becomes antidistinguishable with $(N+1)$-copies [Theorem~\ref{N_N+1}].

Finally, we obtained a complete characterisation of the minimum number of copies required for many-copy antidistinguishability for equiangular pure-state sets. For arbitrary cardinality, this characterisation applies to equiangular states with real inner products, while for three states it extends to the case of complex inner products [Proposition~\ref{prop:three}], illustrated using qubit examples. Our analysis places no restrictions on the allowed measurements; investigating many-copy antidistinguishability under constrained measurement settings is a direction for future work.

\subsection*{Acknowledgment}
DR acknowledges financial support from Council of Scientific and Industrial Research (CSIR) under CSIR Award No.: 09/0015(15326)/2022-EMR-I. TG acknowledges financial support from the ANRF National Post-Doctoral Fellowship (NPDF) under File No. PDF/2025/005147. SS acknowledges financial support from the European Union (ERC StG ETQO, Grant Agreement no.\ 101165230). Views and opinions expressed are however those of the author(s) only and do not necessarily reflect those of the European Union or the European Research Council. Neither the European Union nor the granting authority can be held responsible for them.


\bibliography{multi}

\begin{thebibliography}{30}
\bibitem [{aut()}]{authorComment}%
  \BibitemOpen
  \href@noop {} {}\bibinfo {note} {This follows from the fact that the identity operator cannot be expressed as a linear combination of the projectors onto these states with all coefficients strictly positive [See Proposition \ref{sufficient}].}\BibitemShut {Stop}%
\end{thebibliography}

\end{document}